\documentclass[12pt,cls,onecolumn]{IEEEtran}
\usepackage{graphicx,amsmath,amssymb,epsfig, amsfonts, cite, latexsym, cuted, multicol, multirow, subfigure, stfloats, array, tabularx}
\usepackage{subeqnarray,tabularx}
\usepackage{color}
\usepackage{setspace}
\usepackage{anysize}

\begin{document}

\title{On the Degrees-of-Freedom of the Large-Scale Interfering Two-Way
Relay Network}
\author{\large Hyun Jong Yang, {\em IEEE, Member}, Won-Yong Shin, {\em IEEE, Member}, \\and Bang Chul Jung, {\em IEEE, Senior
Member}
\\
\thanks{This work was support by IITP grant funded by the Korea government (MSIP) (No. B0126-15-1064, Research on Near-Zero Latency Network for 5G Immersive Service).}
\thanks{H. J. Yang is with the School of Electrical and Computer Engineering,
Ulsan National Institute of Science and Technology, Ulsan 689-798,
Republic of Korea. (e-mail: hjyang@unist.ac.kr).}
\thanks{W.-Y. Shin is with the Department of Computer Science and
Engineering, Dankook University, Yongin 448-701, Republic of Korea
(E-mail: wyshin@dankook.ac.kr).}
\thanks{B. C. Jung (corresponding author) is with the Department of Electronics Engineering, Chungnam
National University, Daejeon, Republic of Korea (E-mail:
bcjung@cnu.ac.kr).}
} \maketitle


\markboth{To Appear in IEEE Transactions on Vehicular Technology}
{Yang {\em et al.}: On the Degrees-of-Freedom of the Large-Scale
Interfering Two-Way Relay Network}


\newtheorem{definition}{Definition}
\newtheorem{theorem}{Theorem}
\newtheorem{lemma}{Lemma}
\newtheorem{example}{Example}
\newtheorem{corollary}{Corollary}
\newtheorem{proposition}{Proposition}
\newtheorem{conjecture}{Conjecture}
\newtheorem{remark}{Remark}

\def \diag{\operatornamewithlimits{diag}}
\def \min{\operatornamewithlimits{min}}
\def \max{\operatornamewithlimits{max}}
\def \log{\operatorname{log}}
\def \max{\operatorname{max}}
\def \rank{\operatorname{rank}}
\def \out{\operatorname{out}}
\def \exp{\operatorname{exp}}
\def \arg{\operatorname{arg}}
\def \E{\operatorname{E}}
\def \tr{\operatorname{tr}}
\def \SNR{\operatorname{SNR}}
\def \dB{\operatorname{dB}}
\def \ln{\operatorname{ln}}

\def \bmat{ \begin{bmatrix} }
\def \emat{ \end{bmatrix} }

\def \be {\begin{eqnarray}}
\def \ee {\end{eqnarray}}
\def \ben {\begin{eqnarray*}}
\def \een {\end{eqnarray*}}

\begin{abstract}
Achievable degrees-of-freedom (DoF) of the \textit{large-scale}
interfering two-way relay network is investigated. The network
consists of $K$ pairs of communication nodes (CNs) and $N$ relay
nodes~(RNs). It is assumed that $K\ll N$ and each pair of CNs
communicates with each other through one of the $N$ relay nodes
without a direct link between them. Interference among RNs is also
considered. Assuming local channel state information~(CSI) at each
RN, a distributed and opportunistic RN selection technique is
proposed for the following three promising relaying protocols:
amplify--forward, decode--forward, and compute--forward. As a main
result, the asymptotically achievable DoF is characterized as $N$
increases for the three relaying protocols. In particular, a
sufficient condition on $N$ required to achieve the certain DoF of
the network is analyzed. Through extensive simulations, it is
shown that the proposed RN selection techniques outperform
conventional schemes in terms of achievable rate even in practical
communication scenarios. Note that the proposed technique operates
with a distributed manner and requires only local CSI,
\textcolor{black}{leading to easy implementation }for practical
wireless systems.
\end{abstract}


\newpage

\section{Introduction}


For a three-node relay network with a single pair of communication
nodes (CNs) and a single relay node (RN), two-way relay (TWR)
communication, where relays receive signals from two transmitters
simultaneoulsy and then send signals to the two receivers, doubles
the spectral efficiency of one-way relay (OWR) communications
\cite{B_Rankov07_JSAC,H_Yang11_TIT}. The concept of the TWR
communication has been extended to multi-node interference-limited
relaying networks \cite{H_Yang12_JSAC}. Recently, a combined
technique of network coding and interference alignment (IA) was
adopted to \textit{interfering} TWR networks in order to reduce
the effect of interference
\cite{C_Wang13_Online,Z_Xiang12_Globecom,K_Lee13_TWC,A_Papadogiannis13_CL}.
On the other hand, there have been few schemes that consider a
general interfering TWR network which consists of $K$ pairs of CNs
and $N$ RNs, also known as $K\times N\times K$ interfering TWR
networks. In \cite{B_Rankov07_JSAC}, Rankov and Wittneben showed
that the amplify-and-forward (AF) relaying protocol with
interference-neutralizing beamforming
can achieve the optimal%
\footnote{`Optimal' DoF implies the upper-bound on the DoF of the
channel, which
is usually derived from simple mathematical theorems.%
} DoF of the half-duplex $K\times N\times K$ interfering TWR
network if $N\ge K(K-1)+1$ for a given $K$. However, the scheme in
\cite{B_Rankov07_JSAC} requires global CSI at all nodes and full
collaboration amongst all RNs. The authors of
\cite{T_Gou12_TIT,I_Shomorony13_arXiv} considered the achievable
degrees-of-freedom of $K\times K\times K$ interfering OWR
networks, where the number of CNs and RNs are the same. In
particular, the interference neutralization technique of
\cite{B_Rankov07_JSAC} was combined with the interference
alignment technique to achieve the optimal DoF of the
$2\times2\times2$ interfering OWR network \cite{T_Gou12_TIT} .
However, the scheme in \cite{T_Gou12_TIT} cannot be applied to the
general $K\times N\times K$ interfering TWR network with arbitrary
numbers of $K$ and $N$. In addition, the scheme in
\cite{T_Gou12_TIT} works only with global CSI assumption at each
node.

\textcolor{black}{The internet-of-things (IoT) concept has
recently received much attection from wireless researchers, where
an extremely large number of devices are expected to exist. In
addition, the fifth generation (5G) cellular network is expected
to support more than 10,000 devices, each of which can communicate
directly with others or operate as a relay
\cite{J_Andrews14_JSAC}. Among many devices, a small number of
devices may transmit at a time due to sparse traffic pattern in
the IoT scenario. Several studies have defined and studied the
$(N,K)$-user interference channel ($N\gg K$), in which $K$ user
pairs are selected to communicate at a time
\cite{A_Tajer12_TIT,S_Chae13_TC}. }

\textcolor{black}{In this correspondence, we consider a TWR
network where the number of simultaneously transmitting nodes is
relatively smaller than the number of relaying nodes, which is
referred to as the large-scale interfering TWR network.
Specifically,} we investigate the achievable DoF of the $K\times
N\times K$ interfering TWR network
with local CSI at each node%
\footnote{Each node is assumed to acquire the CSI of its own
incoming or outgoing
channels \cite{K_Gomadam11_TIT}.%
} and without collaboration among nodes in the network.
Three-types of relay protocols are considered: i) AF, ii)
decode--forward (DF), and iii) compute--forward (CF) with lattice
codes. For each source-destination pair, one of $N$ RNs is
selected to help them, and thus, an opportunistic RN selection
(ORS) technique is proposed to mitigate interference. The proposed
ORS technique minimizes the sum of received interference at all
nodes, and thereby maximizes the achievable DoF of the network. We
show that the proposed ORS technique with AF or CF relaying
asymptotically achieves the optimal DoF as the number of RNs, $N$,
increases by rendering the overall network interference-free. In
particular, for given signal-to-noise ratio (SNR) and $K$, we
derive a sufficient condition on $N$ required to achieve the
optimal DoF for AF and CF
relaying, which turns out to be $N=\omega\left(\textrm{SNR}^{2(K-1)}\right)$%
\footnote{The function $f(x)$ defined by $f(x)=\omega(g(x))$
implies that
$\lim_{x\rightarrow\infty}\frac{g(x)}{f(x)}=0$.%
}. On the other hand, it is shown that the DoF with DF relaying is
bounded by half of the optimal DoF. Simulation results show that
the proposed ORS technique outperforms the conventional
max-min-SNR RN selection technique even in practical communication
environments.

\section{System and Channel Models}

\label{SEC:system} Consider the time-division dupex (TDD)
half-duplex $K\times N\times K$ interfering TWR network composed
of $K$ pairs of CNs and $N$ RNs, as depicted in Fig.
\ref{fig:system}. Each pair of the CNs attempts to communicate
with each other through a single selected RN, and no direct paths
between the CNs are assumed, i.e., separated TWR network
\cite{H_Yang11_TIT}. The two sets of CNs at one and the other
sides are referred to as Group 1 and 2, respectively, as shown in
Fig. \ref{fig:system}.

The channel coefficient between the $i$-th CN in Group $n$,
$n\in\{1,2\}$, and RN $j$ is denoted by $h_{n(i),\textrm{R}(j)}$,
$i\in\{1,\ldots,K\}\triangleq\mathcal{K}$,
$j\in\{1,\ldots,N\}\triangleq\mathcal{N}$, assuming TDD channel
reciprocity. It is assumed that each channel coefficient is an
identically and independently distributed (i.i.d.) complex
Gaussian random variable with zero mean and unit variance. In
addition, channel coefficients are assumed to be invariant during
the $T$ time slots, i.e. block fading.

In the first time slot, denoted by Time 1, the CNs transmit their
signals to the RNs simultaneously. In the second time, Time 2, the
selected RNs broadcast their signals to all CNs. The transmit
symbol at the $i$-th CN in Group $n$ in Time 1 is denoted by
$x_{n(i)}$. The maximum average transmit power at the CN is
defined by $P$, and thus the power constraint is given by
\begin{equation}
E|x_{n(i)}|^{2}\le
P,\hspace{10pt}n=1,2.\label{eq:xn_power_constraint}
\end{equation}
Suppose that RN $j$ is selected to serve the $i$-th pair of CNs.
Then, the transmit symbol at RN $i$ is denoted by
$x_{\textrm{R}(j)}$, which includes the information of both
$x_{1(i)}$ and $x_{2(i)}$, and the power constraint is given by
\begin{equation}
E|x_{\textrm{R}(j)}|^{2}\le P.\label{eq:xR_power_constraint}
\end{equation}
That is, the symmetric SNRs are assumed \cite{H_Yang12_JSAC}.

If we denote the achievable rate for transmitting and receiving
$x_{n(i)}$ by $R_{n(i)}$, the total DoF is defined by
\begin{equation}
\textrm{DoF}=\lim_{\textrm{SNR}\rightarrow\infty}\frac{\sum_{i=1}^{K}R_{1(i)}+R_{2(i)}}{\log(\textrm{SNR})},\label{eq:DoF_def}
\end{equation}
where $\textrm{SNR}=P/N_{0}$ and $N_{0}$ is the received noise
variance.

\section{Distributed \& Opportunistic Relay Selection}

\label{sec:relay_selection}

\subsection{Overall Procedure}

\label{subsec:Overall}

\subsubsection{Step 1 - Scheduling Metric Calculation}

From the pilots from the $2K$ CNs in Group 1 and 2, RN $j$,
$j\in\mathcal{N}$, estimates the channels $h_{1(i),\textrm{R}(j)}$
and $h_{2(i),\textrm{R}(j)}$, $i=1,\ldots,K$. Subsequently, RN $j$
calculates the total interference levels (TILs), which account for
the sums of received interference in Time 1 at RN $j$ and leakage
of interference that it generates in Time 2. As seen from Fig.
\ref{fig:system}, the TIL at RN $j$ for the case where it serves
the $i$-th pair of CNs, $i\in\mathcal{K}$, is given by
\begin{equation}
\eta_{i,\textrm{R}(j)}=2\sum_{m=1,m\neq
i}^{K}\left|h_{1(m),\textrm{R}(j),}\right|^{2}+\left|h_{2(m),\textrm{R}(j)}\right|^{2}.\label{eq:eta}
\end{equation}

\subsubsection{Step 2 - RN Selection}

\label{subsubsec:Step2} For the RN selection, we extend the
distributed RN selection algorithm used in \cite{A_Bletsas06_JSAC}
for the OWR network with a single pair of source and destination.

Upon calculating $\eta_{i,\textrm{R}(j)}$, $i=1,\ldots,K$, RN $j$
initiates up to $K$ different back-off timers, which are
respectively proportional to $\eta_{i,\textrm{R}(j)}$, if
$\eta_{i,\textrm{R}(j)}<\epsilon$, where $\epsilon>0$ is the
maximum allowable interference. Specifically, RN $j$ initiates the
back-off timers $\lambda_{i,\textrm{R}(j)}$ given by
\begin{equation}
\lambda_{i,\textrm{R}(j)}=\frac{\eta_{i,\textrm{R}(j)}}{\epsilon}T_{\textrm{max}},
\end{equation}
where $T_{\textrm{max}}$ is the maximum back-off time duration.
After the back-off time $\lambda_{i,\textrm{R}(j)}$, if no RNs
have been assigned to the $i$-th pair of CNs, RN $j$ announces to
serve the $i$-th pair of CNs to all the CNs and RNs in the network
and terminates the selection. Upon acknowledging this
announcement, all other unselected RNs deactivate the timers
corresponding to the $i$-th pair of CNs, i.e.,
$\lambda_{k,\textrm{R}(m)}$, $k\neq j$, $m\in\{\textrm{unselected
RNs}\}$, to exclude the consideration of the selected CNs. In this
way, the RN with the smallest TIL value can be selected in a
distributed fashion for each $i$. Through the proposed RN
selection, we assume without loss of generality that RN $i$ is
selected to serve the $i$-th pair of CNs\textcolor{black}{{} for
notational simplicity}.

Since the RN selection is done only if
$\lambda_{i,\textrm{R}(j)}<\epsilon$, the total time required to
select RNs for all CNs is not greater than $T_{\textrm{max}}$.
Noting that $\eta_{i,\textrm{R}(j)}$ is independent for different
$i$ or $j$ and has a continuous distribution, the probability of a
collision between $\lambda_{i,\textrm{R}(j)}$, $i=1,\ldots,K$'s,
$j=1,\ldots,N$, is arbitrarily small. Thus, $T_{\textrm{max}}$ can
be chosen arbitrarily small compared to the block length $T$. The
efficiency for the achievable rate is lower-bounded by
$\frac{T}{T+T_{\max}}$, which tends to 1 by choosing $T_{\max}$ to
be arbitrarily small compared to $T$ which is relatively large in
general \cite{K_Gomadam11_TIT,I_Shomorony13_arXiv}.

Note that the outage takes place if any RN cannot be assigned for
one or more pairs of CNs because there was no RN with TIL smaller
than $\epsilon$ during the selection process. In the sequel, we
derive a condition on $N$ to make the RN selection always
successful for any given $\epsilon$. In addition, we shall find
practical values of $\epsilon$ for given $N$ through numerical
simulations, which makes the outage probabilities be almost zero.

\subsubsection{Step 3 - Communication}

In Time 1, the CNs transmit their signals to the RNs, and the
received signal at RN $i$ is expressed as
\begin{align}
y_{\textrm{R}(i)} &
=\underbrace{h_{1(i),\textrm{R}(i)}x_{1(i)}+h_{2(i),\textrm{R}(i)}x_{2(i)}}_{\textrm{desired
signal}}+\underbrace{\sum_{k\neq
i,k=1}^{K}\left(h_{1(k),\textrm{R}(i)}x_{1(k)}+h_{2(k),\textrm{R}(i)}x_{2(k)}\right)}_{\triangleq
I_{\textrm{R}(i)},\textrm{interference}}+z_{\textrm{R}(i)},\label{eq:yR}
\end{align}
where $z_{\textrm{R}(i)}$ accounts for the additive white Gaussian
noise (AWGN) at RN $i$ with zero mean and the variance $N_{0}$.
Upon receiving $y_{\textrm{R}(i)}$, RN $i$ generates the transmit
symbol $x_{\textrm{R}(i)}$ from
\begin{equation}
x_{\textrm{R}(i)}=f_{e}(y_{\textrm{R}(i)}),\label{eq:enc_xR}
\end{equation}
where $f_{e}$ is a discrete memoryless encoding function.

In Time 2, RN $i$ then broadcasts $x_{\textrm{R}(i)}$, and the
received signal at the $i$-th CN in Group $n$, $n\in\{1,2\}$, is
written by
\begin{equation}
y_{n(i)}=\underbrace{h_{n(i),\textrm{R}(i)}x_{\textrm{R}(i)}}_{\textrm{desired
signal}}+\underbrace{\sum_{m\neq
i,m=1}^{K}h_{n(i),\textrm{R}(m)}x_{\textrm{R}(m)}}_{\triangleq
I_{n(i)},\textrm{interference}}+z_{n(i)},\label{eq:yCN}
\end{equation}
where $z_{n(i)}$ is the AWGN with zero mean and the variance
$N_{0}$. With the side information of $x_{n(i)}$, the $i$-th CN in
Group $n$ retrieves the symbol transmitted from the other side
from
\begin{equation}
x_{\tilde{n}(i)}=f_{d}(y_{n(i)},x_{n(i)}),\label{eq:dec_xn}
\end{equation}
where $\tilde{n}=3-n$ and $f_{d}$ is a discrete memoryless
decoding function.

The encoding and decoding functions, $f_{e}$ and $f_{d}$,
respectively, differ from relaying protocols, i.e., AF, DF, and
CF. We shall specify them in the sequel in terms of DoF
achievability results. The overall procedure of the proposed
scheme is illustrated in Fig. \ref{fig:overall} for the case of
$K=2$ and $N=3$.

\section{DoF Achievability}

\label{sec:DOF} From (\ref{eq:yR}) and (\ref{eq:yCN}), the sum of
received interference at RN $i$ in Time 1 and at the $i$-th pair
of CNs in Time 2, normalized by the noise variance $N_{0}$, is
expressed as
\begin{align}
\Delta_{i} & \triangleq\frac{E\left|I_{\textrm{R}(i)}\right|^{2}+E\left|I_{1(i)}\right|^{2}+E\left|I_{2(i)}\right|^{2}}{N_{0}}\nonumber \\
{\color{black}} & =\left(\sum_{k\neq
i,k=1}^{K}\left|h_{1(k),\textrm{R}(i)}\right|^{2}+\left|h_{2(k),\textrm{R}(i)}\right|^{2}\right)\textrm{SNR}+\left(\sum_{m\neq
i,m=1}^{K}\left|h_{1(i),\textrm{R}(m)}\right|^{2}+\left|h_{2(i),\textrm{R}(m)}\right|^{2}\right)\textrm{SNR}\label{eq:delta_i}
\end{align}
The following lemma establishes the condition for $N$ required to
decouple the network with constant received interference even for
increasing interference-to-noise-ratio (INR). In particular, even
though there exist a mismatch between the TIL of (\ref{eq:eta})
calculated at RN $i$ with the local CSI and the sum of received
interference in (\ref{eq:delta_i}), we shall show in the proof of
the following lemma that the proposed ORS based on the TIL of
(\ref{eq:eta}) can minimize the sum of received interference at
all nodes, thereby maximizing the achievable DoF.

\begin{lemma}{[}\textbf{Decoupling Principle}{]} \label{lemma:decoupling}
For any $\epsilon>0$, define $\mathcal{P}_{C}$ as
\begin{align}
\mathcal{P}_{C} & \triangleq\textrm{Pr}\left\{ \sum_{i=1}^{K}\Delta_{i}<\epsilon\right\} \label{eq:P_def}\\
 & =\textrm{Pr}\left\{ \sum_{i=1}^{K}\left(E\left|I_{\textrm{R}(i)}\right|^{2}+E\left|I_{1(i)}\right|^{2}+E\left|I_{2(i)}\right|^{2}\right)<\epsilon N_{0}\right\} .
\end{align}
Using the proposed ORS, we have
\begin{align}
\lim_{\textrm{SNR}\rightarrow\infty}\mathcal{P}_{C} &
=1,\label{eq:P_def-1}
\end{align}
if
\begin{equation}
N=\omega\left(\textrm{SNR}^{2(K-1)}\right).\label{eq:N_cond}
\end{equation}

\end{lemma}
\begin{proof}
From the fact that
$\sum_{i=1}^{K}\Delta_{i}=\sum_{i=1}^{K}\eta_{i,\textrm{R}(i)}\textrm{SNR}$,
$\mathcal{P}_{C}$ in the high SNR regime can be rewritten by
\begin{align}
\lim_{\textrm{SNR}\rightarrow\infty}\mathcal{P}_{C} & =\lim_{\textrm{SNR}\rightarrow\infty}\textrm{Pr}\left\{ \sum_{i=1}^{K}\eta_{i,\textrm{R}(i)}\textrm{SNR}<\epsilon\right\} \\
 & \ge\lim_{\textrm{SNR}\rightarrow\infty}\textrm{Pr}\left\{ \eta_{i,\textrm{R}(i)}<\frac{\epsilon\textrm{SNR}^{-1}}{K},\forall i\in\{1,\ldots,K\}\right\} \label{eq:P_D2}\\
 & \ge\lim_{\textrm{SNR}\rightarrow\infty}\left(\textrm{Pr}\left\{ \eta_{i,\textrm{R}(i)}<\frac{\epsilon\textrm{SNR}^{-1}}{K}\right\} \right)^{K},\label{eq:P_D3}
\end{align}
where (\ref{eq:P_D3}) follows from the fact that
$\eta_{i,\textrm{R}(i)}$'s are independent for different $i$.
Since the channel coefficients are independent complex Gaussian
random variables with zero mean and unit variance,
$\frac{\eta_{i,\textrm{R}(i)}}{2}$ is a central Chi-square random
variable with degrees-of-freedom $4(K-1)$. Consequently, the
cumulative density function of $\eta_{i,\textrm{R}(i)}$ is given
by \cite{B_Jung12_TC}
\begin{equation}
F_{\eta}(x)=\frac{\gamma\left(2(K-1),x/4\right)}{\Gamma(2(K-1))},\label{eq:F_bounds}
\end{equation}
where $\Gamma(x)=\int_{0}^{\infty}t^{x-1}e^{-t}dt$ is the Gamma
function and $\gamma(s,x)=\int_{0}^{x}t^{s-1}e^{-t}dt$ is the
lower incomplete Gamma function. In addition, from \cite[Lemma
1]{B_Jung12_TC}, upper and lower bounds on $F_{\eta}(x)$ for
$0<x<2$ are given by
\begin{equation}
C_{1}\cdot x^{2(K-1)}\le F_{\eta}(x)\le C_{2}\cdot
x^{2(K-1)},\label{eq:F_bounds2}
\end{equation}
where
\begin{equation}
C_{1}\triangleq\frac{e^{-1}2^{-4K+3}}{(K-1)\Gamma(2(K-1))}\hspace{5pt}\textrm{
and
}\hspace{5pt}C_{2}\triangleq\frac{2^{-4(K-1)}}{(K-1)\Gamma(2(K-1))}.\label{eq:C1_C2}
\end{equation}
\textcolor{black}{Recall that for notational simplicity, we assume
without loss of generality that RN $i$ is selected to serve the
$i$-th pair of CNs. In addition, let us denote that the $i$-th RN
is selected for the $i$-th pair of CNs in the $\pi(i)$-th
selection, where $\pi(i)\in\left\{ 1,2,\ldots,K\right\} $. Then,}
the probability $\textrm{Pr}\left\{
\eta_{i,\textrm{R}(i)}<\frac{\epsilon\textrm{SNR}^{-1}}{K}\right\}
$ in (\ref{eq:P_D3}) represents the case where at the
\textcolor{black}{$\pi(i)$}-th RN selection, a RN is assigned to
the $i$-th pair of CNs if and only if there exists at least one RN
with the TIL smaller than $\frac{\epsilon\textrm{SNR}^{-1}}{K}$
amongst \textcolor{black}{$\left(N-\pi(i)+1\right)$} unselected
RNs. If we denote the set of indices of the
\textcolor{black}{$\left(N-\pi(i)+1\right)$} unselected RNs at the
\textcolor{black}{$\pi(i)$}-th RN selection by $\mathcal{R}_{i}$,
it follows that \pagebreak[0]
\begin{align}
\pagebreak[0]\textrm{Pr}\left\{ \eta_{i,\textrm{R}(i)}<\frac{\epsilon\textrm{SNR}^{-1}}{K}\right\}  & =1-\textrm{Pr}\left\{ \eta_{i,\textrm{R}(j)}>\frac{\epsilon\textrm{SNR}^{-1}}{K},\forall j\in\mathcal{R}_{i}\right\} \\
 & =1-\left(1-F_{\eta}\left(\frac{\epsilon\textrm{SNR}^{-1}}{K}\right)\right)^{N-{\color{black}\pi(i)}+1}\pagebreak[0]\\
 & \ge1-\left(1-F_{\eta}\left(\frac{\epsilon\textrm{SNR}^{-1}}{K}\right)\right)^{N-K+1}\pagebreak[0]\\
 & \ge1-\frac{\left(1-C_{1}\left(\epsilon/K\right)^{2(K-1)}\cdot\textrm{SNR}^{-2(K-1)}\right)^{N}}{\left(1-C_{2}\left(\epsilon/K\right)^{2(K-1)}\cdot\textrm{SNR}^{-2(K-1)}\right)^{(K-1)}}\pagebreak[0]\label{eq:eta_LB}
\end{align}
where (\ref{eq:eta_LB}) follows from (\ref{eq:F_bounds2}). From
the following Bernoulli's inequality

\begin{equation}
(1-x)^{n}\le\frac{1}{1+nx},\,\, x\in[0,1],\,\, n\in\mathbb{N},
\end{equation}
for sufficiently large SNR to satisfy
$C_{1}\left(\epsilon/K\right)^{2(K-1)}\textrm{SNR}^{-2(K-1)}\le1$,
the last term of \eqref{eq:eta_LB} can be bounded by

\begin{equation}
\frac{\left(1-C_{1}\left(\epsilon/K\right)^{2(K-1)}\cdot\textrm{SNR}^{-2(K-1)}\right)^{N}}{\left(1-C_{2}\left(\epsilon/K\right)^{2(K-1)}\cdot\textrm{SNR}^{-2(K-1)}\right)^{(K-1)}}\le\frac{\left(1-C_{2}\left(\epsilon/K\right)^{2(K-1)}\cdot\textrm{SNR}^{-2(K-1)}\right)^{-(K-1)}}{1+N\cdot
C_{1}\left(\epsilon/K\right)^{2(K-1)}\cdot\textrm{SNR}^{-2(K-1)}}.\label{eq:eq_eta_LB2}
\end{equation}
Therefore, for increasing SNR, the term
$\frac{\left(1-C_{1}\left(\epsilon/K\right)^{2(K-1)}\cdot\textrm{SNR}^{-2(K-1)}\right)^{N}}{\left(1-C_{2}\left(\epsilon/K\right)^{2(K-1)}\cdot\textrm{SNR}^{-2(K-1)}\right)^{(K-1)}}$
tends to 0 if and only if $N\cdot\textrm{SNR}^{-2(K-1)}$ in the
numerator of the right-hand side of \eqref{eq:eq_eta_LB2} tends to
infinity, i.e., $N=\omega\left(\textrm{SNR}^{2(K-1)}\right)$. In
such a case, from \eqref{eq:eta_LB}, we get
\begin{equation}
\lim_{\textrm{SNR
\ensuremath{\rightarrow\infty}}}\textrm{Pr}\left\{
\eta_{i,\textrm{R}(i)}<\frac{\epsilon\textrm{SNR}^{-1}}{K}\right\}
=1.\label{eq:eta_LB3}
\end{equation}
Otherwise, the term
$\frac{\left(1-C_{1}\left(\epsilon/K\right)^{2(K-1)}\cdot\textrm{SNR}^{-2(K-1)}\right)^{N}}{\left(1-C_{2}\left(\epsilon/K\right)^{2(K-1)}\cdot\textrm{SNR}^{-2(K-1)}\right)^{(K-1)}}$
in \eqref{eq:eta_LB} tends to 1 so that $\textrm{Pr}\left\{
\eta_{i,\textrm{R}(i)}<\frac{\epsilon\textrm{SNR}^{-1}}{K}\right\}
$ is unbounded.

From (\ref{eq:P_D3}), (\ref{eq:eta_LB}), and \eqref{eq:eta_LB3},
we have
\begin{align}
\lim_{\textrm{SNR}\rightarrow\infty}\mathcal{P}_{C}\ge\lim_{\textrm{SNR}\rightarrow\infty}\left(\textrm{Pr}\left\{
\eta_{i,\textrm{R}(i)}<\frac{\epsilon\textrm{SNR}^{-1}}{K}\right\}
\right)^{K}=1,\label{eq:P_D_final}
\end{align}
if and only if $N=\omega\left(\textrm{SNR}^{2(K-1)}\right)$ for
any $\epsilon>0$, which proves the lemma.
\end{proof}

\begin{remark}From Lemma \ref{lemma:decoupling}, the $K\times N\times K$
interfering TWR network becomes $K$ isolated TWR networks with
limited interference level even for increasing INR, if
$N=\omega\left(\textrm{SNR}^{2(K-1)}\right)$. In the proposed
scheme, the dimension extension of the time/frequency domain in
the conventional IA technique \cite{S_Jafar08_TIT,K_Gomadam11_TIT}
is replaced by the dimension extension in the number of users.

\end{remark}

Now the following theorem is our main result on the DoF
achievability.

\begin{theorem}\label{theorem:scaling} Using the proposed ORS scheme,
the AF, LC-CF, and DF schemes achieve \textcolor{black}{
\begin{equation}
\textrm{DoF}_{\textrm{AF}}=K,\textrm{\,\,
DoF}_{\textrm{LC-CF}}=K,\textrm{\,\,
DoF}_{\textrm{DF}}=\frac{K}{2},\label{eq:DOF_all}
\end{equation}
}respectively, with high probability if
\begin{equation}
N=\omega\left(\textrm{SNR}^{2(K-1)}\right).\label{eq:N}
\end{equation}
\end{theorem}


Sections \ref{sub:Amplify-and-Forward},
\ref{sub:Lattice-Code-Aided-Compute-and-F}, and
\ref{sub:Decode-and-Forward} \textcolor{black}{prove Theorem
\ref{theorem:scaling} providing detailed encoding and decoding
functions for each scheme}. In addition, Section
\ref{sub:Comparison-of-the} provides comprehensive comparisons
among the AF, LC-DF, and DF schemes in terms of the DoF
achievability.

Note that the overall procedure of the scheduling metric
calculation, RN selection, and communication protocol is analogous
for all the three schemes, and the only difference appears in the
encoding function $f_{e}$ in \eqref{eq:enc_xR} for constructing
$x_{\textrm{R}(i)}$ at the RN and the decoding function $f_{d}$ in
\eqref{eq:dec_xn} for retrieving $x_{1(i)}$ and $x_{2(i)}$ at the
CNs.
%

\subsection{\textcolor{black}{Proof of Theorem \ref{theorem:scaling} for AF}
\label{sub:Amplify-and-Forward}}

In the AF scheme, the relay retransmits the received signal with a
proper amplification. Specifically, from the received signal
$y_{\textrm{R}(i)}$ in (\ref{eq:yR}), RN $i$ generates the
transmit signal $x_{\textrm{R}}$ from
\begin{equation}
x_{\textrm{R}(i)}=\gamma_{i}\cdot
y_{\textrm{R}(i)},\label{x_Ri_AF}
\end{equation}
where $\gamma_{i}>0$ is the amplifying coefficient defined such
that the power constraint \eqref{eq:xR_power_constraint} is met.
Thus, $\gamma_{i}$ can be obtained from

\begin{equation}
\gamma_{i}=\frac{\sqrt{P}}{\sqrt{\sum_{n=1}^{2}|h_{n(i),\textrm{R}(i)}|^{2}P+|I_{\textrm{R}(i)}|^{2}+N_{0}}}.\pagebreak[0]\label{eq:gamma_def}
\end{equation}
Inserting \eqref{x_Ri_AF} into \eqref{eq:yCN} yields the received
signal at the $i$-th CN in Group $\tilde{n}$, $\tilde{n}\in\left\{
0,1\right\} $, given by
\begin{equation}
y_{\tilde{n}(i)}=\gamma_{i}h_{\tilde{n}(i),\textrm{R}(i)}\left(h_{1(i),\textrm{R}(i)}x_{1(i)}+h_{2(i),\textrm{R}(i)}x_{2(i)}+I_{\textrm{R}(i)}+z_{\textrm{R}(i)}\right)+I_{\tilde{n}(i)}+z_{\tilde{n}(i)}.\label{eq:yn_AF}
\end{equation}
The CN then subtracts the known interference signal from
$y_{\tilde{n}(i)}$ to get\pagebreak[0]
\begin{align}
 & y_{\tilde{n}(i)}-\underbrace{\gamma_{i}\cdot h_{\tilde{n}(i),\textrm{R}(i)}h_{\tilde{n}(i),\textrm{R}(i)}x_{\tilde{n}(i)}}_{\textrm{known interference}}\label{eq:y_ni_AF}\\
 & \,\,\,\,\,\,\,\,\,\,\,\,\,\,\,\,=\gamma_{i}h_{\tilde{n}(i),\textrm{R}(i)}h_{n(i),\textrm{R}(i)}x_{n(i)}+\gamma_{i}h_{\tilde{n}(i),\textrm{R}(i)}I_{\textrm{R}(i)}+\gamma_{i}h_{\tilde{n}(i),\textrm{R}(i)}z_{\textrm{R}(i)}+I_{\tilde{n}(i)}+z_{\tilde{n}(i)},\label{eq:y_ni_AF2}
\end{align}
where $\tilde{n}=3-n$. Note here that unlike the DF or LC-CF
scheme, the $i$-th pair of CNs should have the knowledge of the
effective channel $\gamma_{i}\cdot
h_{n(i),\textrm{R}(i)}h_{\tilde{n}(i),\textrm{R}(i)}$.

From \eqref{eq:y_ni_AF2}, the achievable rate for $x_{n(i)}$ is
given by
\begin{align}
R_{n(i)} &
=\frac{1}{2}\log\left(1+\frac{\gamma_{i}^{2}|h_{n(i),\textrm{R}(i)}|^{2}|h_{\tilde{n}(i),\textrm{R}(i)}|^{2}P}{\gamma_{i}^{2}|h_{\tilde{n}(i),\textrm{R}(i)}|^{2}|I_{\textrm{R}(i)}|^{2}+|I_{\tilde{n}(i)}|^{2}+\left(\gamma_{i}^{2}|h_{\tilde{n}(i),\textrm{R}(i)}|^{2}+1\right)N_{0}}\right).
\end{align}
With $N=\omega\left(\textrm{SNR}^{2(K-1)}\right)$, Lemma
\ref{lemma:decoupling} gives us
\begin{equation}
|I_{\textrm{R}(i)}|^{2},|I_{1(i)}|^{2},|I_{2(i)}|^{2}<\epsilon
N_{0}\label{eq:interf_limit}
\end{equation}
for any $\epsilon>0$ with probability $\mathcal{P}_{C}$. Thus, for
any $\epsilon>0$, the{} achievable rate is bounded by

\begin{align}
R_{n(i)} & \ge\mathcal{P}_{C}\cdot\frac{1}{2}\log\left(1+\frac{\gamma_{i}^{2}|h_{n(i),\textrm{R}(i)}|^{2}|h_{\tilde{n}(i),\textrm{R}(i)}|^{2}P}{\left(\gamma_{i}^{2}|h_{\tilde{n}(i),\textrm{R}(i)}|^{2}+1\right)\epsilon N_{0}+\left(\gamma_{i}^{2}|h_{\tilde{n}(i),\textrm{R}(i)}|^{2}+1\right)N_{0}}\right)\label{eq:Rn_AF_final0}\\
{\color{black}} &
=\mathcal{P}_{C}\cdot\frac{1}{2}\log\left(1+\frac{\gamma_{i}^{2}|h_{n(i),\textrm{R}(i)}|^{2}|h_{\tilde{n}(i),\textrm{R}(i)}|^{2}}{\underset{\triangleq
I^{\prime}}{\underbrace{\left(\epsilon+1\right)\left(\gamma_{i}^{2}|h_{\tilde{n}(i),\textrm{R}(i)}|^{2}+1\right)}}}\cdot\frac{P}{N_{0}}\right),\label{eq:Rn_AF_final}
\end{align}
where in \eqref{eq:Rn_AF_final0}, it is assumed that zero rate is
achieved unless the condition $\sum_{i=1}^{K}\Delta_{i}<\epsilon$
holds as in Lemma \ref{lemma:decoupling}. Inserting
\eqref{eq:interf_limit} into \eqref{eq:gamma_def} gives us
\begin{align}
\lim_{\textrm{SNR}\rightarrow\infty}\gamma_{i} &
\ge\lim_{\textrm{SNR}\rightarrow\infty}\frac{1}{\sqrt{\sum_{n=1}^{2}|h_{n(i),\textrm{R}(i)}|^{2}+(\epsilon+1)\textrm{SNR}^{-1}}}=\frac{1}{\sqrt{\sum_{n=1}^{2}|h_{n(i),\textrm{R}(i)}|^{2}}},{\color{black}\pagebreak[0]}\label{eq:R_AF_final-1}
\end{align}
while inserting $\left|I_{\textrm{R}(i)}\right|^{2}=0$ into
\eqref{eq:gamma_def} yields
$\lim_{\textrm{SNR}\rightarrow\infty}\gamma_{i}\le\lim_{\textrm{SNR}\rightarrow\infty}\frac{1}{\sqrt{\sum_{n=1}^{2}|h_{n(i),\textrm{R}(i)}|^{2}+\textrm{SNR}^{-1}}}$.
Thus, we have
$\lim_{\textrm{SNR}\rightarrow\infty}\gamma_{i}=\frac{1}{\sqrt{\sum_{n=1}^{2}|h_{n(i),\textrm{R}(i)}|^{2}}}$
and hence
\begin{align}
\lim_{\textrm{SNR}\rightarrow\infty}I^{\prime} & =\lim_{\textrm{SNR}\rightarrow\infty}\frac{|h_{n(i),\textrm{R}(i)}|^{2}|h_{\tilde{n}(i),\textrm{R}(i)}|^{2}}{\left(\epsilon+1\right)\left(|h_{\tilde{n}(i),\textrm{R}(i)}|^{2}+1/\gamma_{i}^{2}\right)}\label{eq:R_AF_final-1-1}\\
 & =\frac{|h_{n(i),\textrm{R}(i)}|^{2}|h_{\tilde{n}(i),\textrm{R}(i)}|^{2}}{\left(\epsilon+1\right)\left(|h_{\tilde{n}(i),\textrm{R}(i)}|^{2}+\sqrt{\sum_{n=1}^{2}|h_{n(i),\textrm{R}(i)}|^{2}}\right)}\triangleq\hat{I}.{\color{black}\pagebreak[0]}\label{eq:eq:R_AF_final-1-2}
\end{align}

Therefore, the achievable DoF for the AF scheme is given by
\begin{align}
\textrm{DoF}_{\textrm{AF}} & =\lim_{\textrm{SNR}\rightarrow\infty}\frac{\sum_{i=1}^{K}\sum_{n=1}^{2}R_{n(i)}}{\log(\textrm{SNR})}\label{eq:DoF_AF_derivation0}\\
 & \ge\frac{\sum_{i=1}^{K}\sum_{n=1}^{2}\left[\lim_{\textrm{SNR}\rightarrow\infty}\mathcal{P}_{C}\cdot\lim_{\textrm{SNR}\rightarrow\infty}\frac{1}{2}\log\left(1+I^{\prime}\cdot\textrm{SNR}\right)\right]}{\lim_{\textrm{SNR}\rightarrow\infty}\log\textrm{SNR}}\\
 & =\frac{\sum_{i=1}^{K}\sum_{n=1}^{2}1\cdot\lim_{\textrm{SNR}\rightarrow\infty}\frac{1}{2}\log\left(1+I^{\prime}\cdot\textrm{SNR}\right)}{\lim_{\textrm{SNR}\rightarrow\infty}\log\textrm{SNR}}\label{eq:DoF_AF_dervation1}\\
 & =\frac{\sum_{i=1}^{K}\sum_{n=1}^{2}\left[\lim_{\textrm{SNR}\rightarrow\infty}\frac{1}{2}\log\left(\textrm{SNR}\right)+\lim_{\textrm{SNR}\rightarrow\infty}\frac{1}{2}\log\left(\frac{1}{\textrm{SNR}}+I^{\prime}\right)\right]}{\lim_{\textrm{SNR}\rightarrow\infty}\log\textrm{SNR}}\\
 & =\frac{\sum_{i=1}^{K}\sum_{n=1}^{2}\left[\lim_{\textrm{SNR}\rightarrow\infty}\frac{1}{2}\log\left(\textrm{SNR}\right)+\frac{1}{2}\log\left(0+\hat{I}\right)\right]}{\lim_{\textrm{SNR}\rightarrow\infty}\log\textrm{SNR}}\label{eq:eq:DoF_AF_derivation2}\\
 & =K,\label{eq:DoF_AF_derivation_final}
\end{align}
where \eqref{eq:DoF_AF_dervation1} and
\eqref{eq:eq:DoF_AF_derivation2} follow from Lemma
\ref{lemma:decoupling} and \eqref{eq:eq:R_AF_final-1-2},
respectively. On the other hand, the cut-set outer bound
\cite{H_Yang11_TIT}, for which no inter-node interference is
assumed, yields the upper bound $\textrm{DoF}_{\textrm{AF}}\le K$.
Therefore, the achievable DoF with the AF scheme is
$\textrm{DoF}_{\textrm{AF}}=K$, which proves the theorem for
\eqref{eq:DOF_all}.

\subsection{\textcolor{black}{Proof of Theorem \ref{theorem:scaling} for LC-CF}\label{sub:Lattice-Code-Aided-Compute-and-F}}

The LC-CF scheme is a generalized version of the modulo-2 network
coding, in which $x_{1(i)},x_{2(i)}\in\left\{ 0,1\right\} $ and
where $x_{\textrm{R}(i)}=\left[x_{1(i)}+x_{2(i)}\right]_{2}$ is
retransmitted in Time 2. Specifically, in Time 1, $x_{1(i)}$ and
$x_{2(i)}$ are encoded using lattice codes such that
$\left[h_{1(i),\textrm{R}(i)}x_{1(i)}+h_{2(i),\textrm{R}(i)}x_{2(i)}\right]_{\Lambda}$
falls into one of the lattice points in some lattice $\Lambda$.
\textcolor{black}{The encoding functions that generate $x_{1(i)}$
and $x_{2(i)}$ are dependent on the channel coefficients
$h_{1(i),\textrm{R}(i)}$ and $h_{2(i),\textrm{R}(i)}$. Thus, it is
reasonable to assume that the relay designs the encoding functions
and forwards the information on them to the communication nodes,
since the relay can easily acquire $h_{1(i),\textrm{R}(i)}$ and
$h_{2(i),\textrm{R}(i)}$ using the pilot signals transmitted by
the CNs.}

Taking the modulo-$\Lambda$ to the received signal
$y_{\textrm{R}(i)}$ in \eqref{eq:yR}, the RN obtains
\begin{equation}
\left[y_{\textrm{R}(i)}\right]_{\Lambda}=\left[h_{1(i),\textrm{R}(i)}x_{1(i)}+h_{2(i),\textrm{R}(i)}x_{2(i)}+I_{\textrm{R}(i)}+z_{\textrm{R}(i)}\right]_{\Lambda},
\end{equation}
and retrieves the estimate of
$\left[h_{1(i),\textrm{R}(i)}x_{1(i)}+h_{2(i),\textrm{R}(i)}x_{2(i)}\right]_{\Lambda}$
via lattice decoding \cite{H_Yang11_TIT,H_Yang12_JSAC}. More
detailed procedures for constructing $x_{1(i)}$, $x_{2(i)}$, and
$\Lambda$ are omitted, since they are analogous to those for the
three-node TWR channel \cite{H_Yang11_TIT,B_Nazer11_IEEE}, except
that the considered channel includes inter-node interference terms
such as $I_{\textrm{R}(i)}$, $I_{1(i)}$, and $I_{2(i)}$. The RN
then transmits the retrieved signal
$x_{\textrm{R}(i)}=\left[h_{1(i),\textrm{R}(i)}x_{1(i)}+h_{2(i),\textrm{R}(i)}x_{2(i)}\right]_{\Lambda}$,
and then the $i$-th CN in Group $n$ obtains $x_{\tilde{n}(i)}$ in
Time 2 following the two procedures: i) estimating
$x_{\textrm{R}(i)}$ from \eqref{eq:yCN} via lattice decoding, ii)
obtaining $x_{\tilde{n}(i)}$ with known $x_{\textrm{R}(i)}$ and
$x_{n(i)}$ from
$x_{\tilde{n}(i)}=\frac{1}{h_{\tilde{n}(i),\textrm{R}(i)}}\left[x_{\textrm{R}(i)}-h_{n(i),\textrm{R}(i)}x_{n(i)}\right]_{\Lambda}$.

For this lattice encoding and decoding, it is known that the
achievable rates for Time 1 are given by \cite{H_Yang11_TIT}
\begin{align}
R_{n(i)} &
\le\left[\frac{1}{2}\log\left(\tau_{n(i)}+\frac{|h_{n(i),\textrm{R}(i)}|^{2}P}{|I_{\textrm{R}(i)}|^{2}+N_{0}}\right)\right]^{+},\hspace{5pt}n=1,2,\label{eq:R_LC_def}
\end{align}
where $[x]^{+}=\max\{x,0\}$ and
$\tau_{n(i)}\triangleq|h_{n(i),\textrm{R}(i)}|^{2}/\left(|h_{1(i),\textrm{R}(i)}|^{2}+|h_{2(i),\textrm{R}(i)}|^{2}\right)$.
In Time 2, the achievable rate is determined when estimating
$x_{\textrm{R}(i)}$ from \eqref{eq:yCN} \cite{H_Yang11_TIT} as
\begin{align}
R_{n(i)} &
\le\frac{1}{2}\log\left(1+\frac{\left|h_{\tilde{n}(i),\textrm{R}(i)}\right|^{2}P}{\left|I_{\tilde{n}(i)}\right|^{2}+N_{0}}\right).\label{eq:R_GN_def}
\end{align}
With $N=\omega\left(\textrm{SNR}^{2(K-1)}\right)$, Lemma
\ref{lemma:decoupling} gives us
$|I_{\textrm{R}(i)}|^{2},|I_{1(i)}|^{2},|I_{2(i)}|^{2}<\epsilon
N_{0}$ with probability $\mathcal{P}_{C}$. In addition, the
maximum rate of $R_{n(i)}$ is bounded by the minimum of the two
bounds in \eqref{eq:R_LC_def} and \eqref{eq:R_GN_def}. Thus, for
$N=\omega\left(\textrm{SNR}^{2(K-1)}\right)$,
the maximum rate is given by %

\begin{align}
R_{n(i)} & =\min\left\{ \left[\frac{1}{2}\log\left(\tau_{n(i)}+\frac{|h_{n(i),\textrm{R}(i)}|^{2}P}{|I_{\textrm{R}(i)}|^{2}+N_{0}}\right)\right]^{+},\frac{1}{2}\log\left(1+\frac{\left|h_{\tilde{n}(i),\textrm{R}(i)}\right|^{2}P}{\left|I_{\tilde{n}(i)}\right|^{2}+N_{0}}\right)\right\} \label{eq:R_GN_def-1}\\
{\color{black}} & \ge\min\left\{ \mathcal{P}_{C}\cdot\frac{1}{2}\log\left(\tau_{n(i)}+\frac{|h_{n(i),\textrm{R}(i)}|^{2}P}{(1+\epsilon)N_{0}}\right),\mathcal{P}_{C}\cdot\frac{1}{2}\log\left(1+\frac{|h_{\tilde{n}(i),\textrm{R}(i)}|^{2}}{1+\epsilon}\textrm{SNR}\right)\right\} \label{eq:Rn_LC_2-1}\\
 & {\color{black}=}\min\left\{ \mathcal{\mathcal{P}_{C}}\cdot\left(\frac{1}{2}\log(\textrm{SNR})+o_{1}(\textrm{SNR})\right){\color{black},}\mathcal{P}_{C}\cdot\left(\frac{1}{2}\log(\textrm{SNR})+o_{2}(\textrm{SNR})\right)\right\} ,\label{eq:Rn_LC_2_2}
\end{align}
where
$o_{1}(\textrm{SNR})=\frac{1}{2}\log\left(\tau_{n(i)}\textrm{SNR}^{-1}+\frac{|h_{n(i),\textrm{R}(i)}|^{2}}{(1+\epsilon)}\right)$
and
$o_{2}(\textrm{SNR})=\frac{1}{2}\log\left(\textrm{SNR}^{-1}+\frac{|h_{\tilde{n}(i),\textrm{R}(i)}|^{2}}{1+\epsilon}\right)$.

Therefore, with $N=\omega\left(\textrm{SNR}^{2(K-1)}\right)$,
inserting \eqref{eq:Rn_LC_2_2} to (\ref{eq:DoF_def}) and following
the analogous derivation from \eqref{eq:DoF_AF_derivation0} to
\eqref{eq:DoF_AF_derivation_final} give us
$\textrm{DoF}_{\textrm{LC-CF}}=K$, which proves Theorem
\ref{theorem:scaling}.

\begin{remark}\label{remark:subopt_LC}

\textcolor{black}{Optimal lattice coding that achieves Shannon's
capacity bound of $\log(1+SNR)$ may require excessive
computational complexity in the code construction
\cite{U_Erez04_TIT}. Particularly, analytical methods for shaping
the Voronoi region of each lattice point to be a hyper-sphere is
unknown. However, sacrificing this shaping gain by 1.53 dB in SNR,
one can easily design lattice codes with practical non-binary
codes such as low-density parity check codes
\cite{G_Forney98_TIT}, or binary multilevel turbo codes
\cite{U_Wachsmann99_TIT}. For more detailed discussion on the
implementation of lattice codes, the readers are referred to
\cite{H_Yang12_TC} and references therein, or to
\cite{C_Feng13_TIT} and references therein for the effort to
implement practically-tailored lattice codes in two-way relay
channels. }

\end{remark}

\subsection{\textcolor{black}{Proof of Theorem \ref{theorem:scaling} for DF}\label{sub:Decode-and-Forward}}

In the DF scheme, each of $x_{1(i)}$ and $x_{2(i)}$ is
successively decoded at RN $i$ in Time 1 from \eqref{eq:yR}. That
is, $x_{1(i)}$ is decoded first regarding the rest of the terms in
\eqref{eq:yR},
$h_{2(i),\textrm{R}(i)}x_{2(i)}+I_{\textrm{R}(i)}+z_{\textrm{R}(i)}$,
as a noise term, and then is subtracted from $y_{\textrm{R}(i)}$
to decode $x_{2(i)}$. On the other hand, $x_{2(i)}$ can be decoded
first regarding
$h_{1(i),\textrm{R}(i)}x_{1(i)}+I_{\textrm{R}(i)}+z_{\textrm{R}(i)}$
as a noise term, and then subtracted. For this successive
decoding, the rates $R_{1(i)}$ and $R_{2(i)}$ are given by the
multiple-access channel rate bound \cite{B_Rankov07_JSAC} as
follows:
\begin{align}
R_{n(i)} &
\le\frac{1}{2}\log\left(1+\frac{\left|h_{n(i),\textrm{R}(i)}\right|^{2}P}{\left|I_{n(i)}\right|^{2}+N_{0}}\right),\,\,
n=1,2\label{eq:DF_bound1-1}
\end{align}
\begin{align}
R_{1(i)}+R_{2(i)} &
\le\frac{1}{2}\log\left(1+\frac{\left(|h_{1(i),\textrm{R}(i)}|^{2}+|h_{2(i),\textrm{R}(i)}|^{2}\right)P}{\left|I_{\textrm{R}(i)}\right|^{2}+N_{0}}\right).\label{eq:DF_bound2-1}
\end{align}
In Time 2, from individually decoded $x_{1(i)}$ and $x_{2(i)}$,
the network coding is used to construct $x_{\textrm{R}(i)}$ at the
RN as in the LC-CF scheme. Thus, the achievable rates for Time 2
are given again by (\ref{eq:R_GN_def}). Combining
\eqref{eq:DF_bound1-1}, \eqref{eq:DF_bound2-1}, and
\eqref{eq:R_GN_def} together, we obtain the maximum sum-rate as

\begin{align}
R_{1(i)}+R_{2(i)} & =\min\left\{ \sum_{n=1}^{2}\min\left\{
\frac{1}{2}\log\left(1+\frac{\left|h_{n(i),\textrm{R}(i)}\right|^{2}P}{\left|I_{n(i)}\right|^{2}+N_{0}}\right),\frac{1}{2}\log\left(1+\frac{\left|h_{\tilde{n}(i),\textrm{R}(i)}\right|^{2}P}{\left|I_{\tilde{n}(i)}\right|^{2}+N_{0}}\right)\right\}
,\right.\\ \nonumber
 & \left.\,\,\,\,\,\,\,\,\,\,\,\,\,\,\,\,\,\,\,\,\,\,\,\,\,\,\,\,\,\,\,\,\,\,\,\,\,\,\,\frac{1}{2}\log\left(1+\frac{\left(|h_{1(i),\textrm{R}(i)}|^{2}+|h_{2(i),\textrm{R}(i)}|^{2}\right)P}{\left|I_{\textrm{R}(i)}\right|^{2}+N_{0}}\right)\right\} .
\end{align}

From Lemma \ref{lemma:decoupling}, with
$N=\omega\left(\textrm{SNR}^{2(K-1)}\right)$, we have
$|I_{\textrm{R}(i)}|^{2},|I_{1(i)}|^{2},|I_{2(i)}|^{2}<\epsilon
N_{0}$ with probability $\mathcal{P}_{C}$. In such a case, the
maximum sum-rate is bounded by \pagebreak[0]
\begin{align}
R_{1(i)}+R_{2(i)} & \ge\mathcal{P}_{C}\cdot\min\left\{ \sum_{n=1}^{2}\min\left\{ \frac{1}{2}\log\left(1+\frac{\left|h_{n(i),\textrm{R}(i)}\right|^{2}}{\epsilon+1}\textrm{SNR}\right),\frac{1}{2}\log\left(1+\frac{\left|h_{\tilde{n}(i),\textrm{R}(i)}\right|^{2}}{\epsilon+1}\textrm{SNR}\right)\right\} ,\right.{\color{black}\pagebreak[0]}\nonumber \\
 & \left.\,\,\,\,\,\,\,\,\,\,\,\,\,\,\,\,\,\,\,\,\,\,\,\,\,\,\,\,\,\,\,\,\,\,\,\,\,\,\,\underbrace{\frac{1}{2}\log\left(1+\frac{\left(|h_{1(i),\textrm{R}(i)}|^{2}+|h_{2(i),\textrm{R}(i)}|^{2}\right)}{\epsilon+1}\textrm{SNR}\right)}_{{\color{black}\triangleq\Delta_{2}}}\right\} {\color{black}\pagebreak[0]}\label{eq:DF_bound_final0}\\
 & ={\color{black}\mathcal{P}_{C}\cdot\min\left\{ \underbrace{\log\left(1+\frac{\min\left\{ \left|h_{1(i),\textrm{R}(i)}\right|^{2},\left|h_{2(i),\textrm{R}(i)}\right|^{2}\right\} }{\epsilon+1}\textrm{SNR}\right)}_{\triangleq\Delta_{1}},\Delta_{2}\right\} }{\color{black}\pagebreak[0]}
\end{align}
\textcolor{black}{For arbitrarily large SNR and with
$h_{1(i),\textrm{R}(i)}, h_{2(i),\textrm{R}(i)}\neq 0$ , we have
$\Delta_{1}>\Delta_{2}$ since
\begin{equation}
\left(1+\frac{\min\left\{
\left|h_{1(i),\textrm{R}(i)}\right|^{2},\left|h_{2(i),\textrm{R}(i)}\right|^{2}\right\}
}{\epsilon+1}\textrm{SNR}\right)>\left(1+\frac{|h_{1(i),\textrm{R}(i)}|^{2}+|h_{2(i),\textrm{R}(i)}|^{2}}{\epsilon+1}\textrm{SNR}\right)^{1/2}.
\end{equation}
Therefore, for large SNR, the sum-rate can be further expressed
by}
\begin{equation}
{\color{black}R_{1(i)}+R_{2(i)}\ge\mathcal{P}_{C}\cdot\frac{1}{2}\log\left(1+\frac{|h_{1(i),\textrm{R}(i)}|^{2}+|h_{2(i),\textrm{R}(i)}|^{2}}{\epsilon+1}\textrm{SNR}\right).}\label{eq:DF_bound_final}
\end{equation}
Applying \eqref{eq:DF_bound_final} to (\ref{eq:DoF_def}) and
following the analogous derivation from
\eqref{eq:DoF_AF_derivation0} to
\eqref{eq:DoF_AF_derivation_final}, we can only achieve
$\textrm{DoF}_{\textrm{DF}}=K/2$, even under the
interference-limited condition, i.e.,
$N=\omega\left(\textrm{SNR}^{2(K-1)}\right)$.

\subsection{Remark of Theorem \ref{theorem:scaling}: Comparison among the AF,
DF, and LC-CF schemes\label{sub:Comparison-of-the}}

Since the AF scheme only performs power scaling at the RNs, it is
the simplest for \textcolor{black}{implementation but achieves
}the optimal DoF of the network. However, the CN-to-CN effective
channel gain should be known by the CNs, and the scheme suffers
from the noise propagation, particularly in the low SNR regime.
The DF scheme requires the minimum of the CSI, and the
conventional simple coding scheme can be used as in the AF scheme.
Since the noise at the RNs is removed from the decoding at the
RSs, it does not propagate the noise at the RSs. Nevertheless, the
scheme only achieves the half of the optimal DoF. The LC-CF scheme
attains benefits from both AF and DF schemes, i.e., the optimal
DoF and removal of the noise at the RNs through decoding. On the
other hand, the scheme requires lattice encoding and decoding, but
the design of an optimal lattice code for given channel gains
requires an excessive computational complexity
\cite{H_Yang11_TIT}. The suboptimal design of lattice codes can be
considered \textcolor{black}{as discussed in Remark
\ref{remark:subopt_LC}}.

\section{Numerical Examples}

\label{sec:sim} \pagebreak[0] For comparison, two baseline schemes
are considered: max-min-SNR and random selection schemes. In the
max-min-SNR scheme, RN selection is done such that the minimum of
the SNRs of the two channel links between the serving RN and two
CNs is maximized at each selection. \pagebreak[0]


\textcolor{black}{Figure 3 shows the sum-rates versus SNR for
$K=2$, where $N$ increases with respect to SNR according to
Theorem \ref{theorem:scaling}, i.e., $N=\textrm{SNR}^{2(K-1)}$. As
an upper-bound, the sum-rate of the proposed LC-CF ORS scheme but
with no interference is also plotted, the DoF of which is $K$. It
is seen that the proposed AF and LC-CF schemes achieve the DoF of
$K$ as derived in Theorem \ref{theorem:scaling}, whereas the
max-min and random selection schemes achieve zero DoF due to
non-vanishing interference. On the other hand, the DoF of the
proposed DF scheme achieves only $K/2,$ which also complies with
Theorem \ref{theorem:scaling}. It is interesting to note that even
the proposed LC-CF scheme cannot achieve $K$ DoF if $N$ scales
slower than $\textrm{SNR}^{2(K-1)}$, as shown in the example of
the $N=\textrm{SNR}^{(K-1)}$ case which is labeled as `Prop. LC-CF
ORS w/ $N=\textrm{SNR}^{(K-1)}$' in Fig. 3.}

Figure 4 show the sum-rates versus SNR for $K=2$ and (a) $N=20$ or
(b) $N=50$. With fixed and small $N$, the max-min-SNR schemes
outperform the proposed ORS schemes in the low SNR regime, where
the noise is dominant compared to the interference. However, the
sum-rates of the proposed schemes exceed those of the max-min
schemes as the SNR increases, because the interference becomes
dominant than the noise. As a consequence, there exist a crossover
SNR point for each case. As seen from Fig. 4, these crossover
points becomes low as $N$ grows, since the proposed schemes
exploit more benefit as $N$ increases. The proposed schemes
outperform the max-min-SNR schemes for the SNR greater than 7 dB
with $N=50$ as shown in Fig. \ref{fig:rates_SNR_N50}.

Figure \ref{fig:rates_N} shows the sum-rates versus $N$ when $K=2$
and SNR is 20 dB. It is seen that the proposed ORS scheme greatly
enhances the sum-rate of the max-min-SNR scheme for all the cases.
The LC-CF scheme exhibits the highest sum-rates amongst the three
relay schemes for mid-to-large $N$ regime, whereas it slightly
suffers from the rate loss due to $\tau_{n(i)}\le1$ in
(\ref{eq:R_LC_def}) in the small $N$ regime. The sum-rate of the
proposed AF scheme becomes higher than that of the DF scheme as
$N$ increases, because the AF achieves higher DoF, as shown in
Theorem \ref{theorem:scaling}.
%

\clearpage{}

\begin{center}
\begin{figure}
\centering \includegraphics[width=0.7\textwidth]{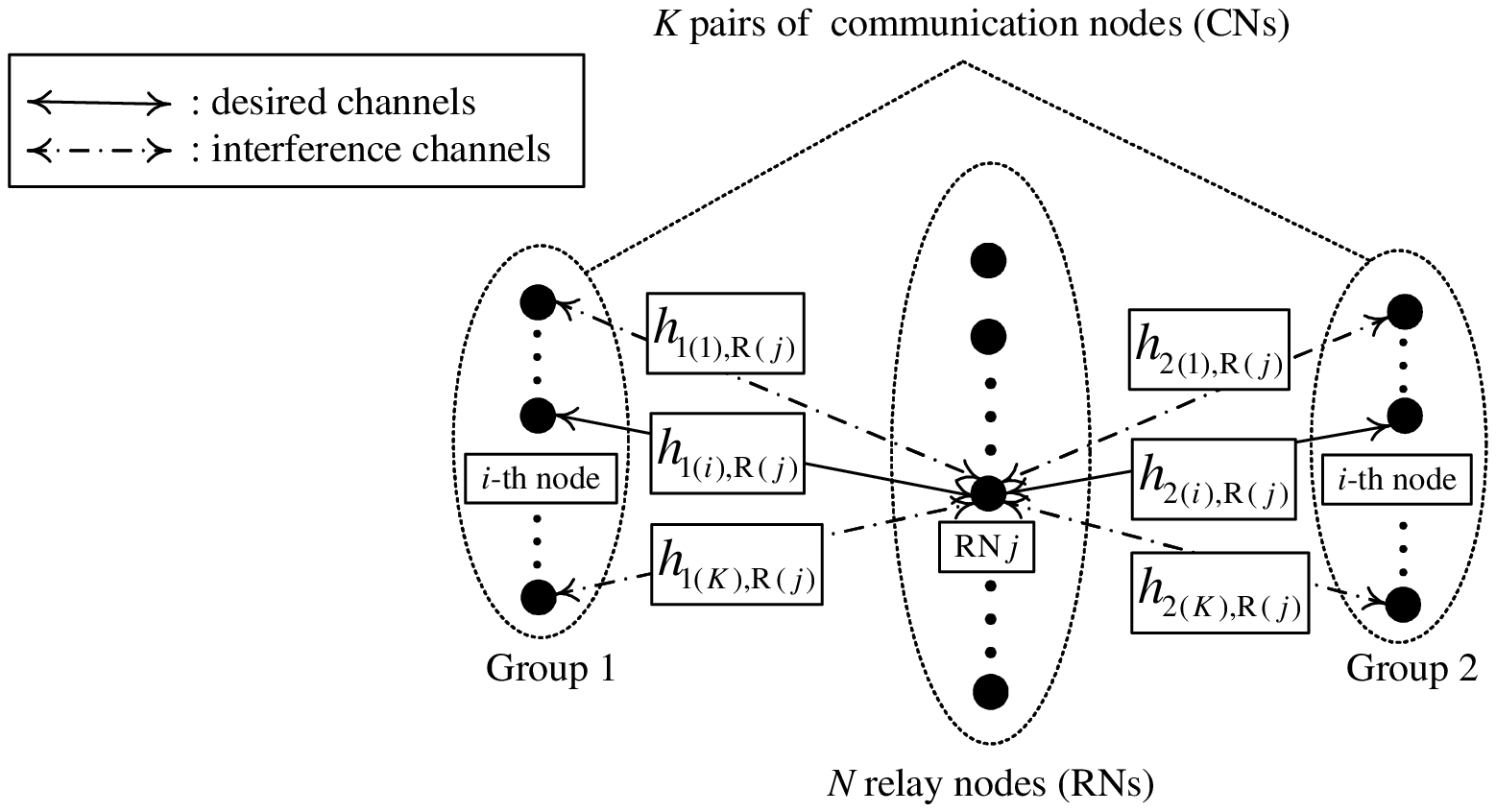}\\
 \protect\protect\caption{The $K\times N\times K$ interfering two-way relay network.}

\label{fig:system}
\end{figure}

\par\end{center}

\begin{center}
\begin{figure}
\centering \includegraphics[width=1\textwidth]{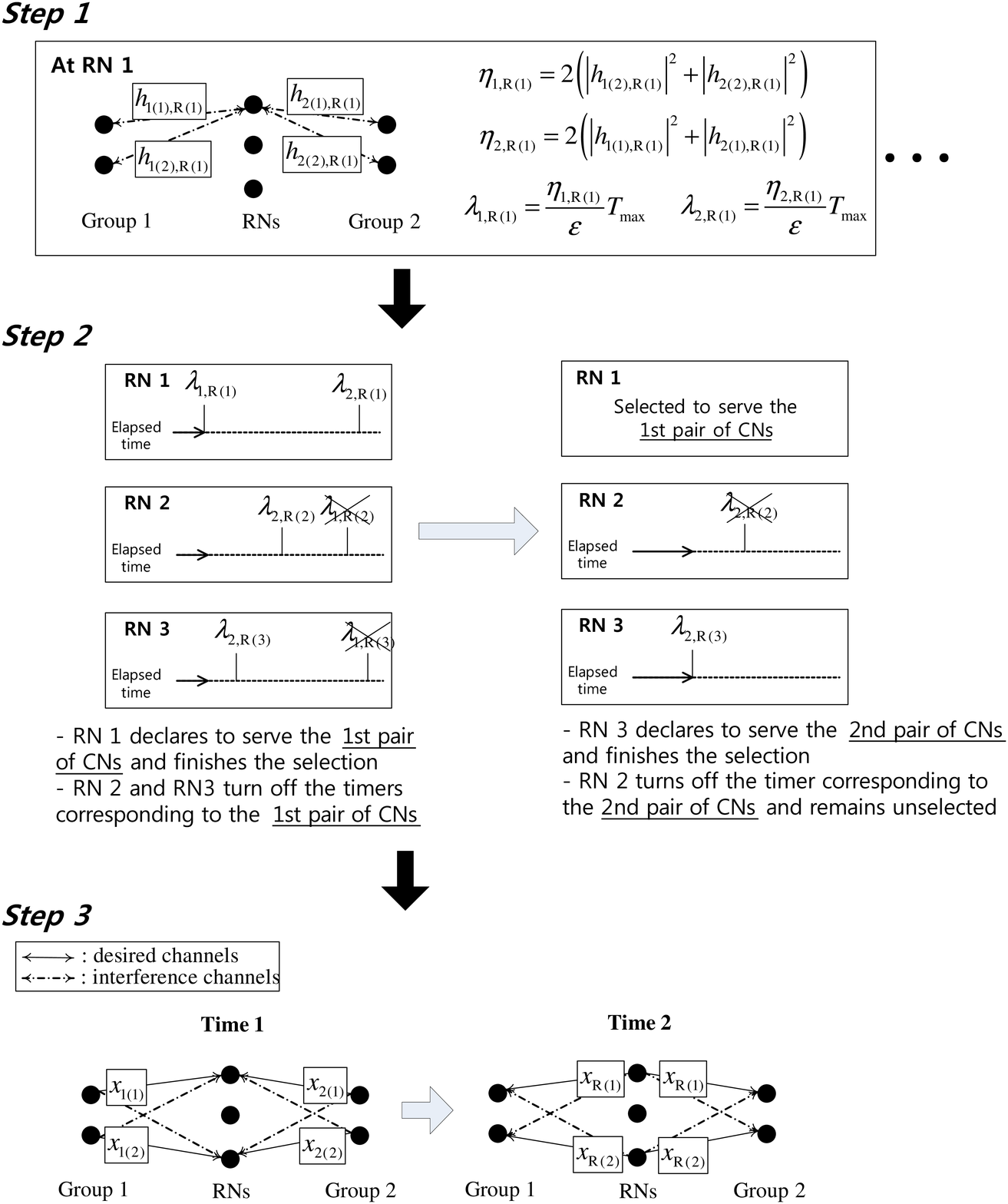}\\
 \protect\protect\caption{Overall procedure of the proposed scheme for $2\times3\times2$ interfering
two-way relay network.}

\label{fig:overall}
\end{figure}

\par\end{center}

\clearpage{}

\begin{center}
%
%
%

\par\end{center}

\begin{center}
\begin{figure}
\begin{centering}
\centering \includegraphics[width=0.7\textwidth]{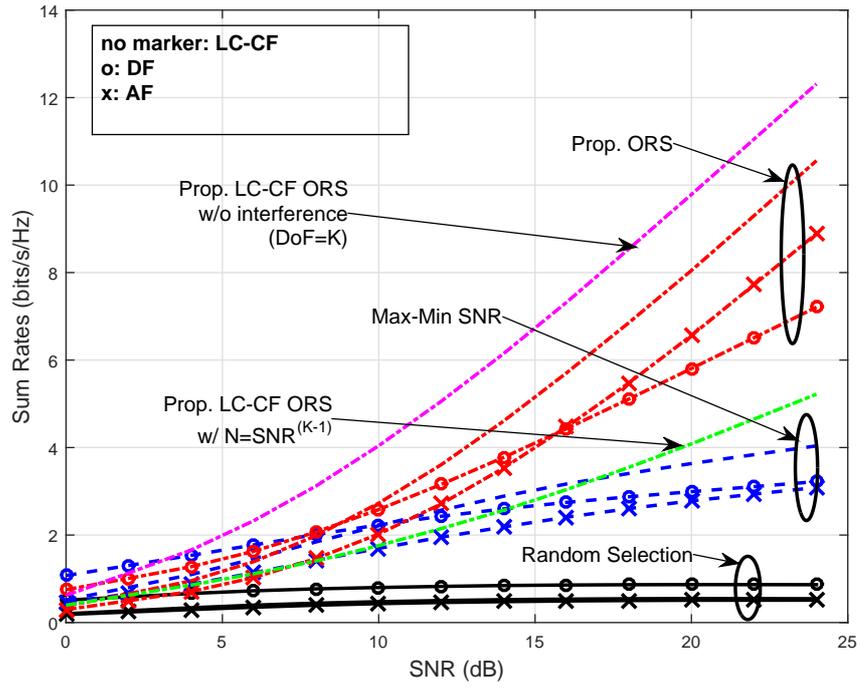}\\
 \protect\protect\caption{Rates versus SNR with $K=2$ and $N=\textrm{SNR}^{2(K-1)}$.}

\par\end{centering}

\label{fig:rate_SNR_N_varying}
\end{figure}

\par\end{center}

\begin{figure}
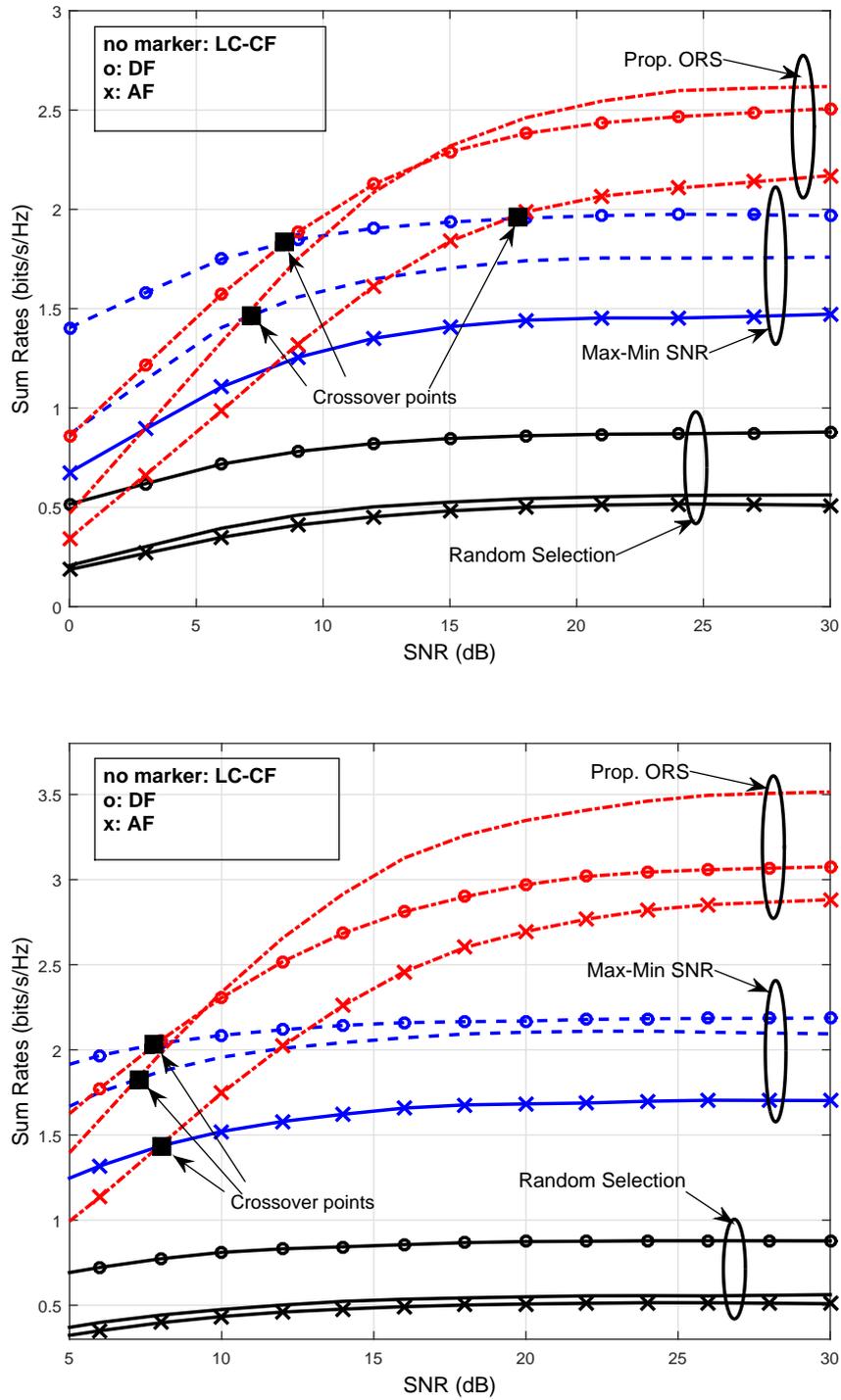

\begin{centering}
\subfigure{
\includegraphics[width=0.7\textwidth]{rates_SNR_N20.eps}\label{fig:rates_SNR_N20}
} \\
\vspace{20pt} \subfigure{
\includegraphics[width=0.7\textwidth]{rates_SNR_N50.eps}
\label{fig:rates_SNR_N50}} \protect\protect\caption{Rates versus
SNR when $K=2$ and (a) $N=20$ or (b) $N=50$.}

\par\end{centering}

\centering{}\label{fig:rates_SNR}
\end{figure}

\clearpage{}

\begin{center}
\begin{figure}
\centering \includegraphics[width=0.7\textwidth]{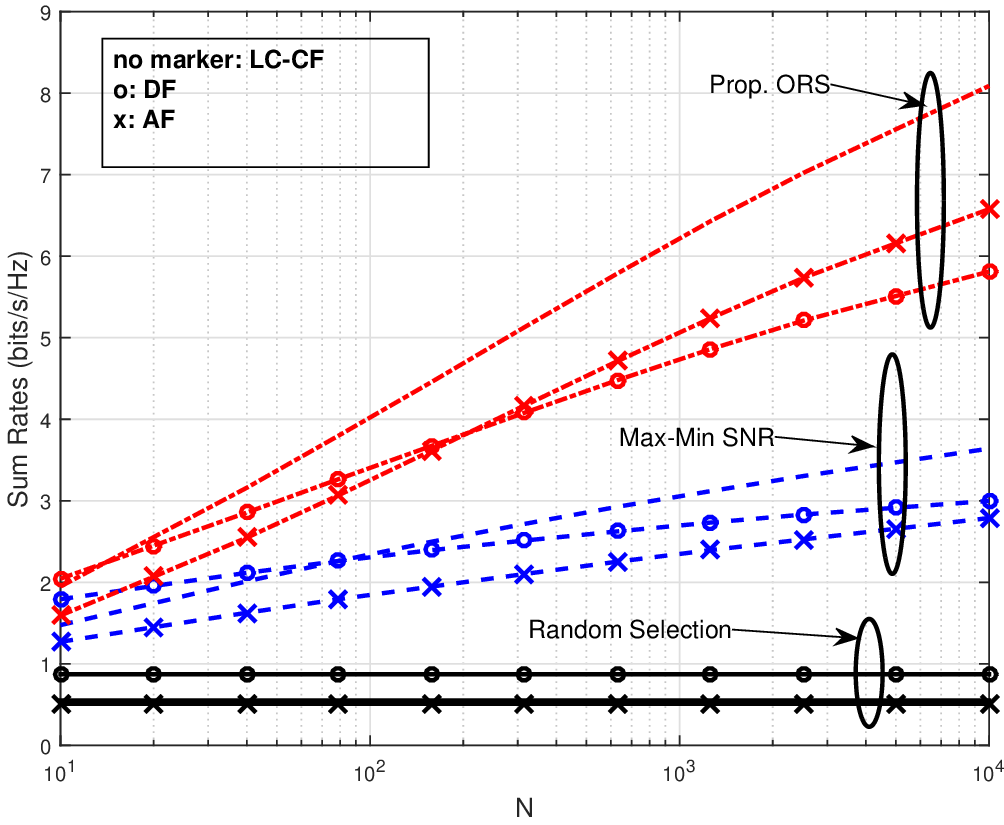}\\
 \protect\protect\caption{Rates versus $N$ when $K=2$ and SNR=20dB. }

\label{fig:rates_N}
\end{figure}

\par\end{center}


\begin{thebibliography}{10}
\providecommand{\url}[1]{#1} \csname url@samestyle\endcsname
\providecommand{\newblock}{\relax}
\providecommand{\bibinfo}[2]{#2}
\providecommand{\BIBentrySTDinterwordspacing}{\spaceskip=0pt\relax}
\providecommand{\BIBentryALTinterwordstretchfactor}{4}
\providecommand{\BIBentryALTinterwordspacing}{\spaceskip=\fontdimen2\font
plus \BIBentryALTinterwordstretchfactor\fontdimen3\font minus
  \fontdimen4\font\relax}
\providecommand{\BIBforeignlanguage}[2]{{%
\expandafter\ifx\csname l@#1\endcsname\relax
\typeout{** WARNING: IEEEtran.bst: No hyphenation pattern has been}%
\typeout{** loaded for the language `#1'. Using the pattern for}%
\typeout{** the default language instead.}%
\else \language=\csname l@#1\endcsname \fi #2}}
\providecommand{\BIBdecl}{\relax} \BIBdecl

\bibitem{B_Rankov07_JSAC}
B.~Rankov and A.~Wittneben, ``Spectral efficient protocols for
half-duplex
  fading relay channels,'' \emph{IEEE J. Select. Areas Commun.}, vol.~25,
  no.~2, pp. 379--389, Feb. 2007.

\bibitem{H_Yang11_TIT}
H.~J. Yang, J.~Chun, and A.~Paulraj, ``Asymptotic capacity of the
separated
  \uppercase{MIMO} two-way relay channel,'' \emph{IEEE Trans. Inform. Theory},
  vol.~57, no.~11, pp. 7542--7554, Nov. 2011.

\bibitem{H_Yang12_JSAC}
H.~J. Yang, Y.~C. Choi, N.~Lee, and A.~Paulraj, ``Achievable
sum-rate of the
  \uppercase{MU-MIMO} two-way relay channel in cellular systems: lattice
  code-aided linear precoding,'' \emph{IEEE J. Selec. Area. Commun.}, vol.~30,
  no.~8, pp. 1304--1318, Sept. 2012.

\bibitem{C_Wang13_Online}
C.~Wang and S.~A. Jafar, ``Degrees of freedom of the two-way relay
  \uppercase{MIMO} interference channel,'' 2013, [Online]. Available at
  http://escholarship.ucop.edu/uc/item/9qc3343h.

\bibitem{Z_Xiang12_Globecom}
Z.~Xiang, J.~Mo, and M.~Tao, ``Degrees of freedom of
\uppercase{MIMO} two-way
  \uppercase{X} relay channel,'' in \emph{IEEE Globecom - Communication Theory
  Symposium}, Anaheim, CA, Dec. 2012, [Online]. Available at
  http://arxiv.org/abs/1208.4048.

\bibitem{K_Lee13_TWC}
K.~Lee, N.~Lee, and I.~Lee, ``Achievable degrees of freedom on
\uppercase{MIMO}
  two-way relay interference channels,'' \emph{IEEE Trans. Wireless Commun.},
  vol.~12, no.~4, pp. 1472--1480, Apr. 2013.

\bibitem{A_Papadogiannis13_CL}
A.~Papadogiannis, A.~G. Burr, and M.~Tao, ``On the maximum
achievable sum-rate
  of interfering two-way relay channels,'' \emph{IEEE Commun. Lett.}, vol.~16,
  no.~1, pp. 72--75, Jan. 2012.

\bibitem{T_Gou12_TIT}
T.~Gou, S.~Jafar, C.~Wang, S.-W. Jeon, and S.-Y. Chung, ``Aligned
interference
  neutralization and the degrees of freedom of the 2 x 2 x 2 interference
  channel,'' \emph{IEEE Trans. Inf. Theory}, vol.~58, no.~7, pp. 4381--4395,
  July 2012.

\bibitem{I_Shomorony13_arXiv}
I.~Shomorony and A.~S. Avestimehr, ``Degrees of freedom of two-hop
wireless
  networks: ``{E}veryone gets the entire cake",'' \emph{IEEE Trans. Inf.
  Theory}, accepted, [Online]. Available at http://arxiv.org/abs/1210.2143.

\bibitem{J_Andrews14_JSAC}
J.~G. Andrews, S.~Buzzi, W.~Choi, S.~V. Hanly, A.~Lozano, A.~C.~K.
Soong, and
  J.~C. Zhang, ``What will {5G be?}'' \emph{IEEE J. Select. Areas Commun.},
  vol.~32, no.~6, pp. 1065--1082, June 2014.

\bibitem{A_Tajer12_TIT}
A.~Tajer and X.~Wang, ``{(n,K)} -user interference channels:
Degrees of
  freedom,'' \emph{IEEE Trans. Inf. Theory}, vol.~58, no.~8, pp. 5338--5353,
  Aug. 2012.

\bibitem{S_Chae13_TC}
S.~H. Chae, B.~C. Jung, and W.~Choi, ``On the achievable
degrees-of-freedom by
  distributed scheduling in an {(N,K)}-user interference channel,'' \emph{IEEE
  Trans. Commun.}, vol.~61, no.~6, pp. 2568--2579, Jun. 2013.

\bibitem{K_Gomadam11_TIT}
K.~Gomadam, V.~R. Cadambe, and S.~A. Jafar, ``A distributed
numerical approach
  to interference alignment and applications to wireless interference
  networks,'' \emph{IEEE Trans. Inf. Theory}, vol.~57, no.~6, pp. 3309--3322,
  June 2011.

\bibitem{A_Bletsas06_JSAC}
A.~Bletsas, A.~Khisti, D.~P. Reed, and A.~Lippman, ``A simple
cooperative
  diversity method based on network path selection,'' \emph{IEEE J. Selec.
  Area. Commun.}, vol.~24, no.~3, pp. 659--672, Mar. 2006.

\bibitem{B_Jung12_TC}
B.~C. Jung, D.~Park, and W.-Y. Shin, ``Opportunistic interference
mitigation
  achieves optimal degrees-of-freedom in wireless multi-cell uplink networks,''
  \emph{IEEE Trans. Commun.}, vol.~60, no.~7, pp. 1935--1944, July 2012.

\bibitem{S_Jafar08_TIT}
S.~A. Jafar and S.~Shamai~(Shitz), ``Degrees of freedom region of
the
  \uppercase{MIMO} \uppercase{X} channel,'' \emph{IEEE Trans. Inf. Theory},
  vol.~54, no.~1, pp. 151--170, Jan. 2008.

\bibitem{B_Nazer11_IEEE}
B.~Nazer and M.~Michael~Gastpar, ``Reliable physical layer network
coding,''
  \emph{Proceedings of IEEE}, vol.~99, no.~3, pp. 438--460, Mar. 2011.

\bibitem{U_Erez04_TIT}
U.~Erez and R.~Zamir, ``Achieving 1/2 log(1 + \uppercase{SNR}) on
the
  \uppercase{AWGN} channel with lattice encoding and decoding,'' \emph{IEEE
  Trans. Inform. Theory}, vol.~50, no.~10, pp. 2293--2314, Oct. 2004.

\bibitem{G_Forney98_TIT}
G.~D.~F. Jr. and G.~Ungerboeck, ``Modulation and codings for
linear gaussian
  channels,'' \emph{IEEE Trans. Inf. Theory}, vol.~44, no.~6, pp. 2384--2415,
  Oct. 1998.

\bibitem{U_Wachsmann99_TIT}
U.~Wachsmann, R.~F.~H. Fischer, and J.~B. Huber, ``Multilevel
codes:
  theoretical concepts and practical design rules,'' \emph{IEEE Trans. Inf.
  Theory}, vol.~45, no.~5, pp. 1361--1391, July 1999.

\bibitem{H_Yang12_TC}
H.~J. Yang, J.~Chun, Y.~Choi, S.~Kim, and A.~Paulraj,
``Codebook-based
  lattice-reduction-aided precoding for limited-feedback coded {MIMO}
  systems,'' \emph{IEEe Trans. Commun.}, vol.~60, no.~2, pp. 510--524, Feb.
  2012.

\bibitem{C_Feng13_TIT}
C.~Feng, D.~Silva, and F.~R. Kschischang, ``An algebraic approach
to
  physical-layer network coding,'' \emph{IEEE Trans. Inf. Theory}, vol.~59,
  no.~11, pp. 7576--7596, Nov. 2013.

\end{thebibliography}
\end{document}